\documentclass[envcountsame]{llncs}
\pdfoutput=1
\usepackage{microtype}
\usepackage{amssymb}
\usepackage{amsmath}
\usepackage{enumerate}
\usepackage{tikz}
\usetikzlibrary{shapes,backgrounds,plothandlers,plotmarks,calc,arrows,fadings}

\usepackage{etoolbox}
\makeatletter
\let\llncs@addcontentsline\addcontentsline
\patchcmd{\maketitle}{\addcontentsline}{\llncs@addcontentsline}{}{}
\patchcmd{\maketitle}{\addcontentsline}{\llncs@addcontentsline}{}{}
\patchcmd{\maketitle}{\addcontentsline}{\llncs@addcontentsline}{}{}
\makeatother
\usepackage{hyperref}
\hypersetup{
    colorlinks,
    linkcolor={red!50!black},
    citecolor={blue!50!black},
    urlcolor={blue!80!black},
    pdftitle={Reconfiguration in bounded bandwidth and treedepth},
    pdfauthor={Marcin Wrochna},
    pdfkeywords={reconfiguration, recoloring, treewidth, bandwidth, treedepth},
}
\usepackage[numbered]{bookmark}

\renewcommand{\qed}{\hfill$\square$\medskip}

\def\thuerule{\leftrightarrow}
\def\reach{\leftrightarrow^*}
\def\tapeL{\mbox{\textdollar}}
\def\tapeR{\mbox{\rlap c{\footnotesize/}}}
\def\tapeE{\mbox{\textvisiblespace}\ }
\def\Oh{\mathcal{O}}

\def\problem#1{\textsc{#1}}
\def\cclass#1{{#1}}

\begin{document}

\title{Reconfiguration in bounded bandwidth and treedepth\texorpdfstring{\thanks{This work is supported by the Foundation for Polish Science (HOMING PLUS/2011-4/8)}}{}}
\date{\small \today}
\author{Marcin Wrochna}
\institute{Uniwersytet Warszawski, Institute of Computer Science, Warsaw, Poland.\\
Email: {\tt mw290715@students.mimuw.edu.pl}}

\maketitle

\begin{abstract}
	We show that several reconfiguration problems known to be \cclass{PSPACE}-complete remain so even when limited to graphs of bounded bandwidth. The essential step is noticing the similarity to very limited string rewriting systems, whose ability to directly simulate Turing Machines is classically known. This resolves a question posed open in [Bonsma P., 2012]. On the other hand, we show that a large class of reconfiguration problems becomes tractable on graphs of bounded treedepth, and that this result is in some sense tight.
\end{abstract}

\section{Introduction}
In the reconfiguration framework one studies how combinatorial objects can be transformed into one another by sequences of small transformations. Usually the set of objects considered is the solution space of a known combinatorial problem and the transformations allowed are changes to a single element of the solution. 
For example, Bonsma~et~al.~\cite{bonsma2009recoloring} studied the problem \problem{$k$-Coloring Reachability}, defined as follows: given two proper $k$-colorings of a graph, can one be transformed into another by changing one color at a time (maintaining a proper coloring throughout).
Another well studied example is independent set reconfiguration \cite{hearn2005ncl,ito2008complexity,kaminski2011shortest}, where given a set of tokens placed on vertices of a graph, one asks whether it is possible to reach another configuration of the tokens by moving one at a time, so that no two tokens are ever adjacent.

Several results suggested that if the graph underlying the combinatorial problem is assumed to have some structure that allows one to find solutions in polynomial time, then questions about reconfiguring solutions can also be answered in polynomial time. For \problem{Independent Set}, deciding the existence of any solution of some size is a classic \cclass{NP}-complete problem. It remains \cclass{NP}-complete even when limited to cubic planar graphs \cite{mohar2001IS_cubic_planar}, but can be solved in polynomial time for bipartite graphs, for claw-free graphs \cite{sbihi1980IS_clawfree,minty1980IS_clawfree,faenza2011IS_clawfree}, and for graphs of bounded treewidth, among others (see \cite{isgci}). In the reconfiguration variants, the reachability problem is known to be \cclass{PSPACE}-complete, even for subcubic planar graphs \cite{hearn2005ncl}. Recently it has been shown to be solvable in polynomial time for claw-free graphs \cite{bonsma2014ISR_clawfree} and cographs~\cite{bonsma2014ISR_cograph}.

Simple algorithms for \problem{Independent Set}, \problem{$k$-Coloring} and many other problems are known for graphs of bounded treewidth (see \cite{bodlaender1994tw_tourist} for an overview and definitions). This motivated the question of determining the complexity of reconfiguration problems in graphs of bounded treewidth, posed open by Bonsma~\cite{bonsma_rerouting_2012}. He further motivated the question by showing that the techniques used for such classes -- dynamic programming -- apply to reconfiguration, by using them to show polynomial algorithms for \problem{Shortest Path Reachability} in planar graphs and for \problem{TAR Reachability} in cographs (i.e., $P_4$-free graphs) \cite{bonsma2014ISR_cograph}. We answer it in the negative, showing that several such problems are \cclass{PSPACE}-complete even when limited to graphs of bounded bandwidth, a notion strictly stronger than treewidth or pathwidth.

To rigorously explore possible patterns in the complexity of reconfiguration problems, we introduce reconfiguration of homomorphisms, or $H$-colorings for digraphs $H$. On one hand, this naturally generalizes \problem{$k$-Coloring Reachability}. On the other hand, this extends the work of Gopalan et~al.~\cite{gopalan2009connectivity} on reconfiguration of generalized satisfiability problems to constraint satisfaction problems (by allowing variables to take more than two different values, but restricting our attention to a single binary relation). Because $H$-colorings provide a special --- though fully expressive --- case of constraint satisfaction problems, we believe them to be the right setting for formally describing patterns arising in reconfiguration problems.

\subsubsection*{Results}
We show that there exist integers $k,b$ such that reachability in reconfiguration variants of \problem{$k$-Coloring}, \problem{Independent Set} and \problem{Shortest Path} is \cclass{PSPACE}-complete even when limited to graphs of bandwidth $b$. As intermediary steps that highlight where the hardness comes from, we show that reconfiguring $k$-list-colorings is \cclass{PSPACE}-complete even for very specific graphs of pathwidth 2, and that there is a digraph $H$ such that reconfiguring $H$-colorings is \cclass{PSPACE}-complete even for paths.

Finally, we give an algorithm for $H$-coloring reconfiguration in graphs of bounded tree-depth. This being very restrictive, the algorithm is not very surprising nor practical, but by connecting the fact with the PSPACE-hardness reductions we show the following: for a class of graphs $\mathcal{C}$ closed under subgraphs, \problem{$H$-Coloring Reachability} problems have polynomial algorithms for all digraphs $H$ if and only if $\mathcal{C}$ has bounded treedepth.
\medskip

Definitions of the problems and graph parameters are, because of their number, only recalled in the section concerning them. For others we refer to the book of Diestel \cite{diestel2000graph}. To omit some technical details, we allow graphs and digraphs to have loops, unless stated otherwise. We don't allow multiple edges, but digraphs can have edges in both directions between two vertices.

\section{String rewriting systems}
The general idea in our reductions is to construct an arbitrarily complicated set of local rules with a fixed instance of the problem -- connecting such instances in a path then allows to simulate the tape of a Turing Machine in a graph of bounded bandwidth. To formalize this into clearly delineated parts we give reductions from the word problem in string rewriting systems (also known as semi-Thue systems and essentially equivalent to unrestricted grammars and finitely presented monoids), whose ability to directly simulate Turing Machines is a well-known, classical result. We construct a very limited \cclass{PSPACE}-complete string rewriting system and later interpret it as an intermediary reconfiguration problem, from which reductions to other problems are easy.

A \emph{string rewriting system} (SRS for short) is a pair $(\Sigma,R)$ where $\Sigma$ is a finite alphabet and $R$ is a set of rules, where each rule is an ordered pair of words $(\alpha,\beta)\in \Sigma^* \times \Sigma^*$. A rule can be applied to a word by replacing one subword by the other, that is, for two words $s,t\in\Sigma^*$, we write $s\rightarrow_R t$ if there is a rule $(\alpha,\beta)\in R$ and words $u,v\in\Sigma^*$ such that $s=u\alpha v$ and $t=u\beta v$. The reflexive transitive closure of this relation defines a reachability relation $\rightarrow^*_R$, where a word $t$ can be reached from another $s$ iff it can be obtained from $s$ by repeated application of rules from $R$. The \emph{word problem} of $R$ is the problem of deciding, given two word $s,t\in\Sigma^*$, whether $s\rightarrow^*_R t$.

A string rewriting system is called \emph{symmetric} when $(\alpha,\beta)\in R \iff (\beta,\alpha)\in R$, in other words, rules are unordered pairs and the reachability relation is symmetric (this is also known as a Thue system). A SRS is called \emph{balanced} if for each rule $(\alpha,\beta)\in R$ we have $|\alpha|=|\beta|$, and $2$-\emph{balanced} if for each rule $(\alpha,\beta)\in R$, $|\alpha|=|\beta|=2$. In a balanced system, only words of the same length can be equivalent. 

The word problem of certain 2-balanced symmetric SRSs is known to be \cclass{PSPACE}-complete. This fact is a folklore variant of the undecidability of general SRSs, whose proof by Emil Post \cite{post1947thue} (and independently by A. A. Markov \cite{markov1947thue}) was described as ``the first unsolvability proof for a problem from classical mathematics''. The essential steps are: encoding the configurations of a Turing machine as a string so that a transition corresponds to string rewriting, padding the encoding so that the corresponding system is balanced, and noticing that the non-reversibility of a TM transition (the asymmetry of the corresponding rewriting system) is not essential in deterministic TMs (see \cite{lewis1980symmetric} for more on symmetric computation). 
An explicit proof of the fact for a balanced symmetric SRS can be found in \cite{bauer_thue_1984} and can easily be adapted to give a 2-balanced symmetric SRS. We include a self-contained proof here for completeness.

\begin{theorem}\label{thm:thue}
There is a 2-balanced symmetric string rewriting system whose word problem is \cclass{PSPACE}-complete (under $\leq^m_{log}$-reducibility).
\end{theorem}
\begin{proof}
Since only words of the same length can be reached by application of rules in a balanced SRS, it suffices to nondeterministically search all words of the same length to solve the problem in nondeterministic polynomial space. By Savitch's theorem \cite{savitch1970}, this places the problem in \cclass{PSPACE}.

Let $M=(\Sigma,Q,q_0,q_{acc},q_{rej},\delta)$ be a deterministic Turing Machine working in space bounded by a polynomial $p(|x|)$ which accepts any \cclass{PSPACE}-complete language. (By starting from a fixed \cclass{PSPACE}-complete problem we show the word problem to be hard for a certain fixed SRS; starting from any language in \cclass{PSPACE} we would only show that the more general word problem, where the SRS is given as input, is \cclass{PSPACE}-complete). $\Sigma$ is the tape alphabet of $M$, $Q$ is the set of states, $q_0,q_{acc},q_{rej}$ are the initial, accepting, and rejecting state respectively, and $\delta: Q\times \Sigma \to Q\times \Sigma \times \{\cdot,L,R\}$ is the transition function of $M$. 
Let $\tapeL,\tapeR\in\Sigma$ denote the left and right end-markers.
Assume w.l.o.g.\ that the machine clears the tape and moves its head to the left end when reaching the accepting state.

For any input $x\in\Sigma^*$ we encode a configuration of the Turing Machine by a word of length exactly $p(|x|)$ over the alphabet $\Gamma = \Sigma \cup (\Sigma \times Q) \cup\{\tapeE\}$. If the tape content is $\tapeL a_1 a_2 \dots a_n \tapeR$ for some $a_1,\dots,a_n\in \Sigma$, the head's position is $i\in\{0,1,\dots,n+1\}$ and the machine's state is $q$, then we define the corresponding word to be the tape content padded with $\tapeE$ symbols and with $a_i$ replaced by $(q,a_i)$, that is $\tapeL{}a_1\dots{}a_{i-1}(q, a_i)a_{i+1}\dots{}a_n\tapeR \tapeE\tapeE \dots \tapeE \in \Gamma^{p(|x|)}$. The initial configuration is then encoded as $s_x=(q_0,\tapeL){x}\tapeR\tapeE\tapeE\dots\tapeE\in \Gamma^{p(|x|)}$ and the only possible accepting configuration is encoded as $t_x=(q_{acc},\tapeL)\tapeR\tapeE\tapeE\dots\tapeE \in \Gamma^{p(|x|)}$. Since $M$ never uses more than $p(|x|)$ space on input $x$, our encoding is well defined for all configurations appearing in the execution of $M$ on $x$. 
So $M$ accepts input $x$ if and only if from $s_x$ one reaches the configuration $t_x$ by repeatedly applying the transition function.
Such an application corresponds exactly to
the following (ordered) string rewriting rules, in the encodings:
\begin{itemize}
\item $\big( (q,a)c\ ,\ (p,b)c \big)$\quad for $q\in Q$,\ $a,c\in\Sigma$ and $\delta(q,a)=(p,b,\cdot)$,
\item $\big( (q,a)c\ ,\ b(p,c) \big)$\quad for $q\in Q$,\ $a,c\in\Sigma$ and $\delta(q,a)=(p,b,R)$,
\item $\big( c(q,a)\ ,\ (p,c)b \big)$\quad for $q\in Q$,\ $a,c\in\Sigma$ and $\delta(q,a)=(p,b,L)$.
\end{itemize}

The transition relation isn't symmetric, but since the machine $M$ is deterministic, the configuration digraph (with machine configurations as vertices and the transition function as the adjacency relation) has out-degree 1. The configuration $t_x$ (which is a configuration in the accepting state) has a loop. Therefore from any configuration, $t_x$ is reachable by a directed path if and only if it is reachable by any path. This means that $M$ accepts input $x$ if and only if applying the transition rules to $s_x$ leads to $t_x$ if and only if $s_x \reach_R t_x$,  where $R$ is the symmetric closure of the above rules, i.e., the 2-balanced symmetric SRS over $\Gamma$ with rules:
\begin{itemize}
\item $\{(q,a)c,(p,b)c\}$\quad for $q\in Q$,\ $a,c\in\Sigma$ and $\delta(q,a)=(p,b,\cdot)$,
\item $\{(q,a)c , b(p,c)\}$\quad for $q\in Q$,\ $a,c\in\Sigma$ and $\delta(q,a)=(p,b,R)$,
\item $\{c(q,a) , (p,c)b\}$\quad for $q\in Q$,\ $a,c\in\Sigma$ and $\delta(q,a)=(p,b,L)$.
\end{itemize}
Since the map $x\mapsto (s_x,t_x)$ is computable in logarithmic space, this proves the world problem of $(\Gamma,R)$ to be \cclass{PSPACE}-hard.
\qed\end{proof}

This result can be slightly strengthened to give a system where only one symbol at a time can be changed. To that aim, it suffices to replace a rule changing two symbols with a sequence of rules using two new intermediary symbols.
\begin{lemma}\label{lem:smallthue}
There is a 2-balanced symmetric SRS $(\Gamma,R)$ whose word problem is PSPACE-complete and such that for every rule $\{a_1a_2,b_1b_2\}\in R$ either $a_1=b_1$ or $a_2=b_2$.
\end{lemma}
\begin{proof}
Let $(\Sigma,R)$ be the 2-balanced symmetric SRS from Theorem~\ref{thm:thue}. Suppose $\{a_1a_2,b_1b_2\}$ is a rule of $R$ in which $a_1\neq b_1$ and $a_2\neq b_2$. We construct a 2-balanced symmetric SRS $(\Gamma,S)$ with one such rule fewer, preserving \cclass{PSPACE}-completeness of the word problem. The claim then follows inductively.

Let $\Gamma=\Sigma\cup\{X,Y\}$ where $X$ and $Y$ are new symbols. Let $S$ be equal to $R$ with rules $\{a_1a_2,Xa_2\},\{X a_2, XY\}, \{XY,b_1 Y\}, \{b_1 Y, b_1 b_2\}$ added and rule $\{a_1a_2,b_1b_2\}$ removed.
We show that for any $s,t\in\Sigma^*$ it holds that $s\reach_R t$ if and only if $s\reach_S t$, which implies that our construction preserves \cclass{PSPACE}-completeness.

Clearly if $s\reach_R t$ then $s\reach_S t$, because replacing $a_1a_2$ with $b_1b_2$ can be done in $S$ by replacing $a_1a_2$ with $Xa_2$, then $XY$, then $b_1 Y$ and finally $b_1b_2$. Suppose now $s\reach_S t$ for some $s,t\in\Sigma^*$. Then there is a sequence $s=u_0,u_1,u_2,\dots,u_l=t$ of words $u_i\in\Gamma^*$ such that $u_i\thuerule_S u_{i+1}$.
Let $\phi : \Gamma^*\to \Sigma^*$ be defined by replacing all $XY$ substrings of a word with $a_1a_2$, then replacing all remaining $X$ symbols with $a_1$ and all remaining $Y$ symbols with $b_2$. It is easy to check that $\phi(u_i)\thuerule_R\phi(u_{i+1})$ or $\phi(u_i)=\phi(u_{i+1})$.
Since $\phi(u_0)=\phi(s)=s$ and $\phi(u_l)=\phi(t)=t$, this implies that $s\reach_R t$. 
\qed\end{proof}
 
\section{A simple intermediary problem}
We define an intermediary problem that highlights how simple a reconfiguration problem achieving \cclass{PSPACE}-hardness can be. Given a pair $H=(\Sigma,E)$, where $\Sigma$ is an alphabet and $E\subseteq \Sigma^2$ a binary relation between symbols, we say that a word over $\Sigma$ is an \emph{$H$-word} if every two consecutive symbols are in the relation (put differently, no element of $\Sigma^2\setminus E$ is a subword). If one looks at $H$ as a digraph (possibly with loops), a word is an $H$-word iff it is a walk in $H$. The \emph{\problem{$H$-Word Reachability}} problem asks whether two given $H$-words of equal length can be transformed into one another by changing one symbol at a time so that all intermediary steps are also $H$-words.

\begin{theorem}\label{thm:hWordReachability}
There is a digraph $H$ for which \problem{$H$-Word Reachability} is \cclass{PSPACE}-complete.
\end{theorem}
\begin{proof}
Let $(\Gamma,S)$ be the 2-balanced symmetric string rewriting system from Lemma~\ref{lem:smallthue} (so if $\{a_1a_2,b_1b_2\}\in S$ then $a_1=b_1$ or $a_2=b_2$). Let $S=\{S_1,\dots,S_m\}$.

Let $\tapeE,\tapeL,\tapeR,x_1,\dots,x_m$ be new symbols, let $\Delta_1 = \{\tapeL,\tapeR,x_1,\dots,x_m\}$, $\Delta_2 = (\Gamma\cup\{\tapeE\})\times(\Gamma\cup\{\tapeE\})$, and let $\Delta = \Delta_1 \cup \Delta_2$. We will call $\Delta_1$ \emph{special symbols} and $\Delta_2$ \emph{pair symbols}. Let $H=(\Delta,E$), where we define $E\subseteq \Delta^2$ as the relation containing the following pairs
\begin{itemize}
\item $((a,b),(b,c))$\quad for any $a,b,c\in \Gamma$,
\item $(\tapeL, (\tapeE,a))$\quad for any $a\in \Gamma$,
\item $((a,\tapeE),\tapeR)$\quad for any $a,b,c\in \Gamma$,
\item $((\cdot,a_1),x_i)$,
\item $((\cdot,b_1),x_i)$,
\item $(x_i, (a_2,\cdot))$,
\item $(x_i, (b_2,\cdot))$ \quad for  any $\cdot\in\Gamma$ and $i\in\{1,\dots,m\}$ such that $S_i=\{a_1a_2,b_1b_2\}$.
\end{itemize}

Let $(s,t)\in \Gamma^*\times \Gamma^*$ be an instance of the word problem for $S$, w.l.o.g. $|s|=|t|=n$. Define $\psi:\Gamma^n\to \Delta^{n+3}$ as 
$$\psi(a_1a_2\dots a_n) = \tapeL (\tapeE,a_1)(a_1,a_2)(a_2,a_3)\dots(a_{n-1},a_n)(a_n,\tapeE) \tapeR$$
It is easy to see that if $s\reach_S t$ then $\psi(s)$ can be transformed into $\psi(t)$, e.g., applying the rule $S_i=\{a_1 a_2, b_1 a_2\}$ corresponds to replacing $(\cdot, a_1)(a_1,a_2)(a_2,\cdot)$ by $(\cdot, a_1)x_i(a_2,\cdot)$, then 
$(\cdot, b_1)x_i(a_2,\cdot)$, then $(\cdot, b_1)(b_1,a_2)(a_2,\cdot)$. We will show the other direction, that if $\psi(s)$ can be transformed into $\psi(t)$, then $s\reach_S t$. Since $\psi$ is computable in logarithmic space, this will imply our claim of \cclass{PSPACE}-completeness.

Indeed, suppose that there is a sequence of $H$-words $\psi(s)=u_0,u_1,\dots,u_l=\psi(t)$ with $u_j\in\Delta^{n+3}$, such that $u_j$ differs from $u_{j+1}$ only at one position. In any $H$-word $v=v_1 v_2\dots v_{n+3}\in\Delta^{n+3}$ there cannot be two consecutive special symbols. We can thus define a word $\phi(v)$ of length $n$ over $\Gamma$ such that its $i$-th symbol, for $i\in\{1,\dots,n\}$, is the second element of $v_{i+1}$ if $v_{i+1}$ is a pair symbol and the first element of $v_{i+2}$ if $v_{i+2}$ is a pair symbol (either case must hold and if both do, the definitions agree by construction of $E$). In particular $\phi(\psi(v))=v$ for any $v\in\Gamma^n$. We argue that $\phi(u_{j-1})\reach_S \phi(u_{j})$ for $j\in\{1,\dots,l\}$.

Notice that the special symbol $\tapeL$ must precede a pair symbol $(\tapeE,\cdot)$ for some $\cdot\in \Gamma$ and any such pair symbol must be preceded by $\tapeL$. Since only one symbol at a time can be changed, it follow inductively that for each $j\in\{0,\dots,l\}$ the first two symbols of $u_j$ must be $\tapeL (\tapeE,\cdot)$ for some $\cdot\in\Gamma$ and $\tapeL$ appears nowhere else. Similarly for the last two symbols, $(\cdot,\tapeE) \tapeR$ for some $\cdot\in\Gamma$.

Since $u_{j-1}$ and $u_j$ differ at only one position, there are non-empty words $v,w\in\Delta^*$ and symbols $a,b\in\Delta$, $a\neq b$ such that $u_{j-1}=vaw$ and $u_j=vbw$. If $a$ or $b$ is a special symbol then both the last symbol of $v$ and the first symbol of $w$ are pair symbols, so $\phi(u_{j-1})=\phi(u_{j})$. Otherwise, let $a=(a_1,a_2), b=(b_1,b_2)$. Assume without loss of generality that $a_1\neq b_1$ and $a_2=b_2$ (the case $a_1=b_1,a_2\neq b_2$ is analogous and the case $a_1\neq b_1,a_2\neq b_2$ can be split by showing that $\phi(u_{j-1})\reach_S \phi(u')$ and $\phi(u') \reach_S \phi(u_j)$ for $u'=v(b_1,a_2)w$, which can easily be checked to be an $H$-word).
If the last symbol of $v$ is a pair symbol $(c,d)$, then $d=a_1$ and $d=b_1$, contradicting our assumption. If the last symbol of $v$ is $\tapeL$, then $a_1=b_1=\tapeE$. Finally if the last symbol of $v$ is $x_i$ for some $i\in\{1,\dots,m\}$, then $S_i$ must be equal $\{ca_1,c'b_1\}$ for some $c,c'\in\Gamma$. Since $a_1\neq b_1$, we have $c=c'$ and the last but one symbol of $v$ must be a pair $(\cdot,c)$ for some $\cdot\in\Gamma\cup\{\tapeE\}$. Thus $\phi(v(b_1,a_2)w)$ is obtained from $\phi(v(a_1,a_2)w)$ by replacing the symbol $a_1$ at position $|v|$, which is preceded by a $c$, by the symbol $b_1$, that is, $\phi(v(b_1,a_2)w) \thuerule_S \phi(v(a_1,a_2)w)$.
\qed\end{proof}

Notice that in the case of \problem{$H$-Word Reachability}, the decision problem asking for the existence of a solution of given length is trivial (even more so since $H$ is fixed), and even counting solutions or extensions of a partial assignment to solutions is easy. This shows that the complexity of a reconfiguration variant of a combinatorial problem can be very different from the complexity of the original problem and its static variants.

\section{Hardness in bounded bandwidth}
In this section we give simple reductions that show several problems studied earlier to be \cclass{PSPACE}-complete in graphs of bounded bandwidth.

The bandwidth of a graph is the minimum over all assignments $f:V(G)\to\mathbb{N}$ of the quantity $\max_{uv\in E(G)} |f(u)-f(v)|$. A graph of bandwidth $b$ can easily be seen to have pathwidth and treewidth at most $b$ (see \cite{schiex1999note}) and maximum degree at most $2b$. On the other hand, the family of stars $K_{1,n}$ gives an example with bounded pathwidth but unbounded bandwidth.
A bucket arrangement of a graph is a partition of the vertex set into a sequence of buckets, such that the endpoints of any edge are either in one bucket or in two consecutive buckets. If a graph has a bucket arrangement where each bucket has at most $b$ vertices, then it has bandwidth at most $2b$ (arrange one bucket after another, with any ordering within one bucket; see also \cite{feige2005bandwidth}).

In \problem{Shortest Path Reachability} one is given a graph $G$ with two distinguished vertices $s,t$ and two shortest paths from $s$ to $t$ ($s-t$ paths). The question is whether one path can be reconfigured into the other by a sequence of shortest $s-t$ paths that differ by one vertex from their predecessor. This provided the first natural example of a \cclass{PSPACE}-hard reconfiguration problem~\cite{bonsma_rerouting_2012} whose underlying combinatorial problem is easy. We prove it remains \cclass{PSPACE}-complete in graphs of bounded bandwidth, by a simple interpretation of Theorem~\ref{thm:hWordReachability}.

\begin{proposition}
There is an integer $b$ such that \problem{Shortest Path Reachability} is \cclass{PSPACE}-complete even when limited to graphs of bandwidth at most $b$.
\end{proposition}
\begin{proof}
Let $H=(\Sigma,R)$ be the graph from  Theorem~\ref{thm:hWordReachability}, we show the lemma for $b=2|\Sigma|$. Let $s,t\in\Sigma^*$ be an instance of \problem{$H$-Word Reachability} with $|s|=|t|=n$.
We construct an instance $(G_n, P_s,P_t)$ of \problem{Shortest Path Reachability} as follows.
The graph $G_n$ depends only on $n$ (and the fixed graph $H$).
Its vertex set contains $v_0,v_{n+1}$ and vertices $v_i^a$ for all $i\in\{1,\dots,n\}$ and $a\in\Sigma$. Its edge set contains $v_0 v_1^a$ and $v_n^a v_{n+1}$ for all $a\in \Sigma$, and $v_i^a v_{i+1}^b$ for all $(a,b)\in R$ and $i\in\{1,\dots,n-1\}$. Let $V_i = \{v_i^a \mid a\in\Sigma\}, V_0=\{v_0\}, V_{n+1}=\{v_{n+1}\}$. The sets $V_i$ give a bucket arrangement, so the bandwidth of $G_n$ is at most $2|\Sigma|$.

A shortest path from $v_0$ to $v_{n+1}$ must go through exactly one vertex in each set $V_i$ and thus defines a word. It is easy to see that this defines an bijection between $H$-words and shortest paths. We let $P_s,P_t$ be the paths corresponding to $s,t$. Changing one symbol corresponds to changing one vertex of the path, so the instances are clearly equivalent.
\qed\end{proof}

We prove the same result for \problem{Maximum Independent Set Reachability}. In this problem one is given a graph $G$ and two maximum independent sets in it. The question is whether one can be reconfigured into the other by a sequence of \emph{token jumps} -- moves that consist of removing a vertex and adding another. Since a maximum independent set is maintained throughout, the vertex added must be adjacent to the one removed (otherwise the set with both added would be independent and larger). Our hardness results thus applies also to \emph{token sliding} and \emph{token addition removal} models (see \cite{kaminski2011shortest} for definitions and some equivalences). The graph constructed in the proof can be obtained simply by taking complements of edge sets on each bucket pair $V_i\cup V_{i+1}$ from the previous construction.
\begin{proposition}
There is an integer $b$ such that \problem{Maximum Independent Set Reachability} is \cclass{PSPACE}-complete even when limited to graphs of bandwidth at most $b$.
\end{proposition}
\begin{proof}
Let $H=(\Sigma,R)$ be the graph from  Theorem~\ref{thm:hWordReachability}, let $p=2|\Sigma|$ and let $s,t\in\Sigma$ be an instance of \problem{$H$-Word Reachability} with $|s|=|t|=n$.
We construct an instance $(G_n, I_s,I_t)$ of \problem{Maximum Independent Set Reachability} as follows. The vertex set of $G_n$ contains vertices $v_i^a$ for all $i\in\{1,\dots,n\}$ and $a\in\Sigma$. Let $V_i = \{v_i^a \mid a\in\Sigma\}$ for $i\in\{1,\dots,n\}$. The edge set of $G_n$ contains an edge between every two vertices of $V_i$ for $i\in\{1,\dots,n\}$ and an edge $v_i^a v_{i+1}^b$ for all $(a,b)\not\in R$ and $i\in\{1,\dots,n-1\}$. The sets $V_i$ give a bucket arrangement, so the bandwidth of $G_n$ is at most $2|\Sigma|$.

An independent set of size $n$ contains exactly one vertex in each clique $V_i$ and thus defines a word. This clearly gives a bijection between $H$-words and maximum independent sets. We let $I_s,I_t$ be the paths corresponding to $s,t$. Since tokens can only move between vertices of one clique, changing one symbol of the corresponding word, the instances are equivalent.
\qed\end{proof}

In \problem{$k$-List-Coloring Reachability} one is given a graph whose vertices are labeled with lists $l(v) \subseteq \{1,\dots, k\}$ and two $k$-list colorings $s,t$ ($k$-colorings of the graph that obey the lists, that is, $s(v)\in l(v)$ for all $v\in V(G)$). The problem asks whether one can be transformed into the other by changing one color at a time and maintaining a proper $k$-list-coloring throughout. A \emph{chain of onions of width $b$ and length $n$} is the graph with vertex set $\{u_1,\dots,u_n\}\cup \{v_i^j \mid i=1,\dots,n-1,\ j=1,\dots,b\}$ and edges $u_i v_i^j $ and $v_i^j u_{i+1}$ for $i\in\{1,\dots,n-1\},j\in\{1,\dots,b\}$. Clearly such graphs have bandwidth $b$ (and pathwidth 2).
\begin{proposition}\label{lem:listColoring}
There are integers $k,b$ such that \problem{$k$-List-Coloring Reachability} is \cclass{PSPACE}-complete even when limited to chains of onions of width $b$.
\end{proposition}
\begin{proof}
Let $H=(\Sigma,R)$ be the graph from  Theorem~\ref{thm:hWordReachability}, let $b=|\Sigma^2\setminus R|$, $k=2|\Sigma|$ and let $s,t\in\Sigma$ be an instance of \problem{$H$-Word Reachability} with $|s|=|t|=n$. We construct the following instance of \problem{$k$-List-Coloring Reachability}. The graph will be the chain of onions of width $b$ and length $n$ with vertices named as above. Let $\Sigma'$ be a disjoint copy of $\Sigma$, we write $a'$ for the copy of $a\in\Sigma$, the set of available colors will be $\Sigma \cup \Sigma'$. The list of a vertex $u_i$ is $\Sigma$ for $i$ even and $\Sigma'$ for $i$ odd. For each $(a,b)\in \Sigma^2 \setminus R$ we choose a different $j\in\{1,\dots,|\Sigma^2\setminus R|\}$ and let $v_i^j$ have the list $\{a,b'\}$ for $i$ even and $\{a',b\}$ for $i$ odd.

A coloring of vertices $u_i$ defines a word in $\Sigma^n$ (by letting the $i$-th symbol be $a$ if the color of $u_i$ is $a$ or $a'$). If $(a,b)\in\Sigma^2 \setminus R$, then for any even $i$, vertices $u_i,u_{i+1}$, with lists $\Sigma,\Sigma'$ respectively,  are adjacent to some vertex $v_i^j$ with list $\{a,b'\}$. Thus in any proper coloring it cannot  be that $u_i$ has color $a$ and $u_{i+1}$ has color $b'$. Similarly for odd indices, thus $ab$ is never a subword of a word corresponding to a proper coloring. Since this holds for all $(a,b)\in \Sigma^2 \setminus R$, such a word is an $H$-word. Conversely, for any $H$-word, the corresponding coloring of $u_i$ can easily be extended to a proper coloring of all the graph. Changing one color corresponds to changing at most one symbol of the corresponding $H$-word, and changing one symbol corresponds to changing the colors of one vertex $u_i$ and of those vertices adjacent to it that would have the same color. This can be done while maintaining a proper coloring, so the instances are equivalent.
\qed\end{proof}

\begin{proposition}
\label{propo:kcoloring}
There are integers $k,b$ such that \problem{$k$-Coloring Reachability} is \cclass{PSPACE}-complete even when limited to graphs of bandwidth at most $b$.
\end{proposition}
\begin{proof} Let $k,b$ be integers from Proposition~\ref{lem:listColoring} and let $b'=b+k$. An instance of \problem{$k$-List-Coloring Reachability} of bandwidth $b$ can be transformed into an equivalent instance of \problem{$k$-Coloring Reachability} of bandwidth at most $b'$ simply by adding a clique of size $k$ to the graph for each original vertex, assigning to their vertices all the colors in some order and replacing the list of each original vertex by edges to vertices of its clique that have colors outside of the list. The colors of any $k$-clique clearly can never be changed in the new instance, and an extension of the cliques' coloring to a coloring of the new instance is proper if and only if it is a proper list-coloring of the original instance.
\qed\end{proof}

\section{Homomorphism reconfiguration on paths, trees and cycles}



In this section we describe a different view making the intermediary problem interesting in its own right. Given two digraphs $G,H$, an $H$-homomorphism of $G$ is a function from vertices of $G$ to vertices of $H$ such that arcs of $G$ are mapped to arcs of $H$ in the same direction. Such an assignment of vertices of $H$ (colors) to vertices of $G$ is also called an $H$-coloring, since it generalizes proper $k$-colorings: a $K_k$-homomorphism of a graph (where undirected edges are equivalent to arcs in both directions) is the same as a proper $k$-coloring. It is a major open problem whether for each digraph $H$, the problem of deciding the existence of an $H$-homomorphisms of a digraph is \cclass{P} or \cclass{NP}-complete. Such a dichotomy is known for undirected graphs \cite{hell1990undir_homomorphism_dichotomy}. These problems are a special case of Constraint Satisfaction Problems, but it is known that a dichotomy for digraph homomorphisms would imply a dichotomy for general CSPs \cite{feder1998homo_digraph}. See \cite{hell2004homomorphism_book} for an overview of these and related results.

Extending the work of Gopalan et al. on the reconfiguration of generalized SAT problems to Constraint Satisfaction Problems, we introduce the following problem. In \problem{$H$-Coloring Reachability} one asks whether in a given graph two $H$-colorings can be transformed into one another by changing one color (i.e., the mapping of one vertex) at a time, maintaining an $H$-coloring throughout. Notice that an $H$-coloring of a directed path is the same as an $H$-word (colors correspond to symbols). Therefore Theorem~\ref{thm:hWordReachability} can be reformulated as follows.
\begin{corollary}\label{cor:digraphRecoloring}
There is a digraph $H$ for which \problem{$H$-Coloring Reachability} is \cclass{PSPACE}-complete even on directed paths.
\end{corollary}

For an undirected graph $H$, reconfiguring $H$-colorings a path turns out to be computationally easy. The idea is to reach a coloring $ababa\dots$, for some edge $ab$ of $H$ in quadratically many steps, and then only moving between such alternating colorings. This is easily generalized to trees, as in the following proposition.
\begin{proposition}\label{propo:Hrecoloring_on_trees}
Let $H$ be a graph (possibly with loops) and let $\alpha,\beta$ be two $H$-colorings of a tree $T$ with root $r$ (chosen arbitrarily) and at least 2 vertices. Then $\alpha$ can be reconfigured to $\beta$ if and only if there is in $H$ a walk of even length  between the colors assigned to $r$ in the two colorings.
\end{proposition}
\begin{proof}
Let $T_i$ be the vertices at distance exactly $i$ from $r$ in the tree $T$.
We first prove inductively that for $i=0,\dots,n-1$, any $H$-coloring $\gamma$ of $T$ can be reconfigured into a coloring $\gamma_i$ such that
\begin{itemize} 
\item $\gamma$ and $\gamma_i$ assign the same colors to each vertex in $T_0,T_1,\dots, T_{n-i}$.
\item for every $v\in T_j$ with $j\geq n-i+1$, $\gamma_i$ assigns the same color to $v$ and its grandparent.
\end{itemize}
For $i=0$ it suffices to take $\gamma_0=\gamma$. For $i+1$, it suffices to take the coloring $\gamma_i$ and recolor one by one each vertex $v$ in $T_{n-i},T_{n-i+2},T_{n-i+4}, \dots$ (in that order) to the color of $v$'s grandparent -- since all of $v$'s sons have the same color as $v$'s parent, which is a color adjacent in $H$ to the color of $v$'s grandparent, this remains a valid $H$-coloring.

Let $S=T_0\cup T_2\cup \dots$ and $S'=T_1\cup T_3\cup \dots$. In $\alpha_{n-1}$ all vertices in $S$ have the same color, say $a$, and also all vertices in $S'$ have the same color, say $a'$. Similarly for $\beta_{n-1}$ with colors $b,b'$. If there is an even-length walk $a_0,a_1,\dots,a_{2l}$ in $H$ with $a_0=a$ and $a_{2l}=b$, then all $S'$ can be recolored from $a'$ to $a_1$, $S$ from $a$ to $a_2$, $S'$ to $a_3$ and so on until $S'$ is recolored $a_{2l-1}$ and $S$ is recolored $a_{2l}$, after which $S'$ can be recolored to $b'$ to get $\beta_{n-1}$.

If, on the other hand, there is no even-length walk from $a$ to $b$, then they must belong to different components of $H$ or to different sides of a bipartition of a component of $H$. Since every vertex has some neighbor, the colors of vertices in $T$ must remain in the same component of $H$ and if the component is bipartite, the colors of $S$ must be and remain on one side, while the colors of $S'$ remain on the other side. In either case, the color $a$ cannot be changed to $b$.
\qed\end{proof}

Nevertheless, \problem{$H$-Coloring Reachability} for an undirected graph $H$ can still be hard on very simple graphs. We prove this for cycles, but the following proof can easily be adapted to give other graphs, like paths with a triangle attached to each end.

\begin{proposition}\label{propo:Hrecoloring_on_cycles}
There is a graph $H$ for which \problem{$H$-Coloring Reachability} is \cclass{PSPACE}-complete even on cycles.
\end{proposition}
\begin{proof}
Let $H=(\Sigma,E)$ be the digraph from Corollary~\ref{cor:digraphRecoloring}. The construction can be easily adapted so that \problem{$H$-Coloring Reachability} is \cclass{PSPACE}-complete on directed cycles of length divisible by 3. Let $H'=(\Sigma', E')$ be a graph with $\Sigma' = \Sigma \times \{0,1,2\}$ and $E'=\{ \{(a,i),(b,i+1\mod 3)\}\ |\ (a,b) \in E \}$. For an $H$-coloring $\alpha$ of a directed cycle $v_1,\dots,v_{3n},v_1$, define $\alpha'(v_i) = (\alpha(v_i), i \mod 3)$ to be an $H'$-coloring of the underlying undirected cycle. The second element of each pair color $\alpha'(v_i)$ cannot ever be changed (the projection to the second elements gives a `frozen' 3-coloring). The relation constraining the first elements is hence exactly $E$, the direction being implied by the 3-coloring. Therefore if $\alpha,\beta$ are two $H$-colorings of the directed cycle, then one can be recolored into the other if and only if $\alpha'$ can be recolored into $\beta'$ as $H'$-colorings of the undirected cycle.
\qed\end{proof}

\section{Treedepth}
\def\td{\mbox{td}}
Treedepth is another parameter describing sparse graphs. Introduced by Ne\v{s}et\v{r}il and Ossona de Mendez \cite{nevsetvril2006tree}, it found various algorithmic applications. It is defined as the minimum height of a rooted forest closure $F$ such that the graph is a subgraph of $F$. A rooted forest closure of height $h$ is any graph obtained from a disjoint union of rooted trees of height $h$ by adding edges from every vertex to all of its ancestors.
There are only finitely many graphs of bounded treedepth with no non-trivial automorphisms. Intuitively, such graphs can be large only by having many copies of the same subgraph, all acting the same way. This allows to reduce many problems by merging the copies, which is made formal by the following theorem. 

\begin{theorem}[\cite{nevsetvril2006tree}]
\label{thm:treedepth}
	For all integers $N$ and $t$, there is an integer $\digamma(N,t)$ such that
	for any graph $G$ of treedepth at most $t$ and any mapping $g:V(G)\to \{1,\dots,N\}$,
	there is a subset $A$ of $V(G)$ of cardinality at most $\digamma(N,t)$ such that
	$G$ has a $g$-preserving homomorphism to $G[A]$.
\end{theorem}

If one defines the treedepth of a digraph to be the treedepth of the underlying graph, the proof of Theorem~\ref{thm:treedepth} can easily be extended to digraphs (giving homomorphisms that preserve edge directions). Seeing that the proof is constructive, such a homomorphism can be computed in time polynomial in $\digamma(N,t)\cdot n$. In the following, $H$ refers to any digraph, and may be given as part of the input.

\begin{proposition}
	\problem{$H$-Coloring Reachability} on loopless digraphs of treedepth at most $t$ can be solved in time $\Oh(f(|H|,t) n^{\Oh(1)})$ for some $f$.
\end{proposition}
\begin{proof}
Let $G$ be a loopless digraph of treedepth $t$ and let $\alpha,\beta$ be two $H$-colorings.
Let $N=|V(H)|^2$ and define $g(v) = (\alpha(v),\beta(v))$. By Theorem~\ref{thm:treedepth}
there is an $\alpha$- and $\beta$-preserving homomorphism $\mu$ from $G$ to some induced subdigraph $G[A]$ of cardinality at most $\digamma(N,t)$.
We claim that $\alpha$ can be reconfigured to $\beta$ in $G$ if and only if $\alpha|_A$  can be reconfigured to $\beta|_A$ in $G[A]$. The latter can be checked by brute force after finding the homomorphism, giving the claimed algorithm.

Let $\alpha_0,\alpha_1,\alpha_2,\dots,\alpha_l=\beta$ be a reconfiguration sequence from $\alpha_0=\alpha$ to $\alpha_l=\beta$ in $G$. Then the sequence $\alpha_i|_A$ is a reconfiguration sequence in $G[A]$.

Let $\alpha_0,\alpha_1,\alpha_2,\dots,\alpha_l=\beta$ be a reconfiguration sequence from $\alpha_0=\alpha$ to $\alpha_l=\beta$ in $G[A]$. 
Define mappings $\alpha'_i:V(G)\to H$ as  $\alpha'_i(v) := \alpha_i(\mu(v))$.
Suppose $\alpha'_i$ isn't a proper coloring -- then there is an edge $vw\in E(G)$ such that $\alpha'_i(v)\alpha'_i(w)\not\in E(H)$, that is, $\alpha_i(\mu(v))\alpha_i(\mu(w))\not\in E(H)$. But $\mu$ is a homomorphism, so $\mu(v)\mu(w)$ is an edge of $G[A]$, contradicting that $\alpha_i$ is a proper $H$-coloring.

Since $\mu$ is $\alpha$ and $\beta$-preserving, we have $\alpha'_0=\alpha$ and $\alpha'_l=\beta$. It suffices to show that $\alpha'_i$ can be reconfigured to $\alpha'_{i+1}$ for all $i$. Let $v$ be the only vertex of $A$ on which $\alpha_i$ and $\alpha_{i+1}$ differ. Then $\alpha'_i$ differs from $\alpha'_{i+1}$ only on $\mu^{-1}(v)$, where one coloring assigns the color $\alpha_i(v)$ and the other assigns $\alpha_{i+1}(v)$ to all vertices. Since all of these vertices map to $v$ through a homomorphism and there is no loop $vv$ in $E(G)$, $\mu^{-1}(v)$ is an independent set of $G$. Therefore, $\alpha'_i$ can be reconfigured to $\alpha'_{i+1}$ by changing the colors of vertices in $\mu^{-1}(v)$ from $\alpha_i(v)$ to $\alpha_{i+1}(v)$ one by one, in any order.
\qed\end{proof}
\pagebreak[3]

Note that a class of graphs has unbounded treedepth if and only if it has graphs with arbitrarily long (undirected) paths as subgraphs \cite{nevsetvril2008grad}. 
Therefore, by Corollary~\ref{cor:digraphRecoloring}, if a class $\mathcal{C}$ of graphs closed under taking subgraphs has unbounded treedepth, then for some digraph $H$, \problem{$H$-Coloring Reachability} is \cclass{PSPACE}-complete on orientations of $\mathcal{C}$ (the class of digraphs obtained by orienting edges of a graph in $\mathcal{C}$). 
This suggests that to find any interesting regularities in the complexity of reachability problems, limiting the structure formed by constraints between variables is not enough, the constraints themselves have to be very simple.

\begin{theorem}
Let $H^*$ be the digraph from Corollary~\ref{cor:digraphRecoloring} and assume \cclass{P}$\neq$\cclass{PSPACE}. For any class $\mathcal{C}$ of graphs without loops closed under taking subgraphs, the following statements are equivalent: 
\begin{itemize}
\item $\mathcal{C}$ has bounded treedepth.
\item \problem{$H$-Coloring Reachability} on orientations of $\mathcal{C}$ is in \cclass{P} for any digraph~$H$.
\item \problem{$H^*$-Coloring Reachability} on orientations of $\mathcal{C}$ is in \cclass{P}.
\end{itemize}
\end{theorem}

It is well known that much weaker conditions are needed for the underlying problems: \problem{$H$-Coloring} is solvable in polynomial time (for every fixed $H$) for any class of bounded treewidth. See \cite{grohe2007complexity} for a tight characterization.


\section{Final remarks}
Not all natural reachability problems are hard in graphs of bounded treewidth. Notice that \problem{Clique Reachability} (in any model where adjacency of configurations can be tested in polynomial time) becomes trivial, since a clique must be contained in one bag of a tree decomposition. There are at most $2^{t+1} n$ different cliques in a graph of treewidth $t$ and they can be all enumerated in time $\Oh(2^{t+1} n)$. 

Our results should also be contrasted with positive results about the reconfiguration of colorings in graphs of bounded treewidth. In particular Dyer et~al.~\cite{dyer2006randomly} proved that for any graph of treewidth $t$ (or even any graph of degeneracy $t$) and any $k\geq t+2$, all $k$-colorings can be reached from one another (see also \cite{cereceda2007mixing,bonamy2013recoloring}). In Proposition~\ref{propo:kcoloring} the bandwidth, though constant, is strictly larger than the number of colors.

The specific value of bandwidth for which our reductions show hardness depend polynomially on the size (alphabet size times the number of states) of the PSPACE-complete Turing machine we start from. There are surprisingly small undecidable Turing machines and (non-balanced) string rewriting systems (e.g.~\cite{tseitin1958small_thue}), but we have been unable to find any similar work about small PSPACE-hard balanced systems or linearly bounded Turing machines. Let us only mention that instances with exponentially long solutions for \problem{Shortest Path Reachability} in graphs of bandwidth 13 are explicitly given in \cite{kaminski2011shortest}, by a~simple construction that may be seen as an illustrative case of some of the ideas presented here.

This article only analyzed the complexity of the reachability problem between two solutions. While the techniques can be easily adapted to answer some related questions (like the complexity of reaching any solution containing a given element or the diameter of solution graphs) others (like the complexity of deciding the connectivity of the solution graph) seem for now elusive.

\bibliographystyle{splncs}
\bibliography{pathwidth}

\end{document}